\documentclass[11pt,a4paper]{article}
\usepackage{authblk} 
\usepackage{fancyhdr} 
\usepackage[a4paper, total={6in, 8in}]{geometry}

\usepackage[T1]{fontenc}

\usepackage{amsmath,amssymb,amsthm}
\usepackage{mathtools}
\usepackage[capitalise]{cleveref}
\usepackage{comment}
\usepackage{proof}
\usepackage{graphicx}
\usepackage{tikz}
\usetikzlibrary{shapes,shapes.geometric,arrows,fit,calc,positioning,automata,chains,matrix.skeleton, arrows.meta}
\usepackage[caption=false]{subfig}

\newtheorem{theorem}{Theorem}
\newtheorem{lemma}{Lemma}
\newtheorem{example}{Example}
\newtheorem{remark}{Remark}

\usepackage[noend,linesnumbered]{algorithm2e}
\RestyleAlgo{ruled}
\SetKwComment{Comment}{/* }{ */}
\SetKwInOut{Input}{Inputs}
\SetKwInOut{Output}{Output}
\Crefname{algocf}{Algorithm}{Algorithms} 
\SetKw{Continue}{continue}

\newcommand{\Set}[1]{\{#1\}}

\newcommand{\sem}[1]{[\![#1]\!]}
\renewcommand{\vec}[1]{\boldsymbol{#1}}
\newcommand{\supp}[1]{\mathop{\mathrm{supp}(#1)}}
\renewcommand{\deg}[1]{\mathop{\mathrm{deg}(#1)}}
\newcommand{\coeff}[1]{\mathop{\mathrm{coeff}(#1)}}
\newcommand{\depth}[1]{\mathop{\mathrm{depth}(#1)}}
\newcommand{\reps}[1]{\left\lVert #1 \right\rVert}
\newcommand{\moninc}[3]{#1\mathrm{\uparrow}^{#2}_{#3}}

\tikzstyle{io} = [trapezium, 
trapezium stretches=true,
trapezium left angle=70, 
trapezium right angle=110, 
minimum width=3cm, 
minimum height=1cm, text centered, 
draw=black, fill=blue!30]
\tikzstyle{process} = [rectangle, 
minimum width=3cm, 
minimum height=1cm, 
text centered, 
text width=3cm, 
draw=black, 
fill=orange!30]
\tikzstyle{arrow} = [thick,->,>=stealth]

\DeclareRobustCommand{\good}{\Simley{0.5}{0.2}\xspace}
\DeclareRobustCommand{\bad}{\Simley{-0.5}{0.2}\xspace}
\newcommand{\Simley}[2]{%
	\begin{tikzpicture}[scale=#2]
	\newcommand*{\SmileyRadius}{1.0}%
	;
	
	\pgfmathsetmacro{\eyeX}{0.5*\SmileyRadius*cos(30)}
	\pgfmathsetmacro{\eyeY}{0.5*\SmileyRadius*sin(30)}
	\draw [line width=0.25mm] (\eyeX-0.25,\eyeY) -- (\eyeX-0.25,\eyeY+0.375);
	\draw [line width=0.25mm] (-\eyeX+0.25,\eyeY) -- (-\eyeX+0.25,\eyeY+0.375);
	
	\pgfmathsetmacro{\xScale}{2*\eyeX/180}
	\pgfmathsetmacro{\yScale}{1.0*\eyeY}
	\draw[line width=0.25mm, domain=-\eyeX:\eyeX]
	plot ({\x},{
		-0.1+#1*0.15 
		-#1*1.75*\yScale*(sin((\x+\eyeX)/\xScale))-\eyeY});
	\end{tikzpicture}%
}%

\providecommand{\keywords}[1]
{
  \small	
  \textbf{\textit{Keywords---}} #1
}

\title{Analyzing Value Functions of States in Parametric Markov Chains\thanks{Work supported by the Flemish (Belgium) inter-university (iBOF) ``DESCARTES'' and the Belgian F.R.S.-FNRS/FWO SynthEx (G0AH524N) projects.}}

\author{Kasper Engelen}
\author{Guillermo A. P\'erez}
\author{Shrisha Rao}

\affil{University of Antwerp -- Flanders Make, Antwerpen, Belgium\\
\texttt{\{kasper.engelen, guillermo.perez, shrisha.rao\}@uantwerpen.be}}
\date{}

\begin{document}

\maketitle

\begin{abstract}

Parametric Markov chains (pMC) are used to model probabilistic systems with unknown or partially known probabilities. Although (universal) pMC verification for reachability properties is known to be coETR-complete, there have been efforts to approach it using potentially easier-to-check properties such as asking whether the pMC is monotonic in certain parameters. In this paper, we first reduce monotonicity to asking whether the reachability probability from a given state is never less than that of another given state. Recent results for the latter property imply an efficient algorithm to collapse same-value equivalence classes, which in turn preserves verification results and monotonicity. We implement our algorithm to collapse ``trivial'' equivalence classes in the pMC and show empirical evidence for the following: First, the collapse gives reductions in size for some existing benchmarks and significant reductions on some custom benchmarks; Second, the collapse speeds up existing algorithms to check monotonicity and parameter lifting, and hence can be used as a fast pre-processing step in practice. 

\vspace{10pt}
\keywords{Markov chains, State-space reduction, Parameters.}
\end{abstract}

\section{Introduction}

Finite-state Markovian models are widely used as an operational model for the quantitative analysis of systems with probabilistic behavior. For state reachability in Markov chains, there are well known probabilistic model checkers such as Storm \cite{DBLP:conf/cav/DehnertJK017} and PRISM \cite{DBLP:conf/cav/KwiatkowskaNP11}. However, in practice, it is not guaranteed that the exact probabilities in the model are known. One way to circumvent this is by modeling unknown probabilities as parameters and known relations between probabilities as functions over shared parameters.

Parametric Markov chains (pMCs, for short) \cite{DBLP:conf/ictac/Daws04} extend classical Markov chains with parameters, i.e. a finite set of real-valued variables, so that transition probabilities are now polynomials over these parameters. Applications of pMCs include model repair \cite{DBLP:conf/tacas/BartocciGKRS11,DBLP:journals/iandc/Chatzieleftheriou18,DBLP:conf/tase/ChenHHKQ013,DBLP:journals/pe/GoubermanST19,DBLP:conf/nfm/PathakAJTK15}, planning in POMDPs~\cite{DBLP:conf/vmcai/HeckSJMK22}, and optimizing randomized distributed algorithms \cite{DBLP:conf/srds/AflakiVBKS17}. Analysis of pMCs is supported by established probabilistic model checkers Storm, PRISM, as well as dedicated tools such as PARAM~\cite{DBLP:conf/cav/HahnHWZ10} and PROPhESY~\cite{DBLP:conf/cav/DehnertJJCVBKA15}.

The pMC verification problem asks: \textit{Given $\lambda\in[0,1]$, is the probability of reaching the target state from the start state, for all valid parameter valuations, at least $\lambda$?} This problem is known to be coETR-complete \cite{DBLP:conf/gd/Schaefer09,DBLP:conf/atva/EngelenPR23}.
Related to the verification problem, the concept of the never-worse relation (NWR) for pMCs has been introduced in \cite{DBLP:conf/fossacs/0001P18,DBLP:conf/atva/EngelenPR23}. Intuitively, a state is never worse than another if, for all valuations of the parameters, the probability of reaching the target state from the former state is at least as much as that from the latter. Another related notion is that of monotonicity: The (value from a) state is said to be monotonically increasing in a parameter if, for all valuations of the parameters, the probability of reaching the target state from it increases when the value of the parameter increases. Both the NWR and monotonicity turn out to be coETR-complete:
Spel et al. \cite{DBLP:conf/tacas/SpelJK21} give a simple polynomial-time reduction from pMC verification to the monotonicity problem by first reducing it to the NWR problem and then reducing the NWR problem to the monotonicity problem. Since the NWR problem is coNP-hard even when parameterization is trivial (i.e., every transition is labeled with a unique variable) \cite{DBLP:conf/fossacs/0001P18}, this implies that the monotonicity problem is also coNP-hard when the parameters are trivial.

Despite the discouraging complexity lower bounds recalled above, efficiently-checkable sufficient conditions to establish monotonicity were given in~\cite{DBLP:conf/tacas/SpelJK21}. Similarly, efficiently-checkable sufficient conditions for establishing never-worse relationships were studied in~\cite{DBLP:conf/atva/EngelenPR23}. In particular, an efficient algorithm to compute equivalence classes of trivially parameterized MCs also follows from \cite{DBLP:conf/atva/EngelenPR23}.  

\paragraph{Contributions.}
In this work, we first give a quasi-polynomial-time construction of a pMC whose solution function is the derivative of the solution function of a given pMC. Interestingly, our construction relies on classical results from circuit complexity theory that also allow us to give a polynomial-sized representation of solution functions. Having access to this derivative pMC, we establish a reduction from monotonicity to the NWR problem. Then, we give an algorithm for collapsing some (not necessarily all) equivalence classes of the pMC using the techniques from \cite{DBLP:conf/atva/EngelenPR23}. We present an implementation in Storm and perform experiments. We evaluate our implementation on several benchmarks from the existing literature and some custom benchmarks, showing that the equivalence-class collapse runs very fast, even for larger benchmarks, and speeds up the monotonicity check in many cases. Finally, we also evaluate our tool as a pre-processing step before parameter lifting (a technique used for near-optimal parameter synthesis)~\cite{DBLP:conf/tacas/SpelJK21} and observe that there too the equivalence-class collapse is sometimes useful.

\section{Preliminaries}
In this section, we introduce the notation we will be using for polynomials,
circuits, and Markov chains with parameters. For the latter, we follow as much as possible the notation used in~\cite{DBLP:journals/iandc/BaierHHJKK20}.
Let $x_1,\dots,x_m$ be \emph{parameters} that can take real values and write
$\vec{x} = (x_1, \dots, x_m)$.

\subsection{Polynomials and rational functions}
By $\mathbb{Q}[\vec{x}]$, we denote the
\emph{polynomial ring} in $\vec{x}$ over the rationals; by
$\mathbb{Q}(\vec{x})$, the analogue \emph{rational field}. For an (exponent)
vector $\vec{\alpha} \in \mathbb{N}^m$, we write $\vec{x}^{\vec{\alpha}}$ to
denote the monomial $\prod_{i=1}^m x_i^{\alpha_i}$. Now, every
\emph{polynomial} $f \in \mathbb{Q}[\vec{x}]$ can be written as a sum of
products
\(
  \sum_{\vec{\alpha} \in I} c_{\vec{\alpha}} \vec{x}^{\vec{\alpha}}
\)
for $I$ a finite subset of $\mathbb{N}^m$ and where $c_{\vec{\alpha}} \in
\mathbb{Q}$ for all $\vec{\alpha} \in I$. The \emph{support} of $f$ is the
subset of $I$ with a nonzero coefficient, $\supp{f} = \Set{\vec{\alpha}
\in \mathbb{N}^m \mid c_{\vec{\alpha}} \neq 0}$. When $\supp{f}$ is empty, we say $f$
is the zero function and write $f \equiv 0$. The \emph{degree} and
\emph{maximal coefficient component} of $f$ are
$\deg{f} = \max \Set{\sum_{i=1}^m \alpha_i \mid \vec{\alpha} \in \supp{f}}$
and $\coeff{f} = \max\Set{ |a|,|b| \in \mathbb{N} : \vec{\alpha} \in \supp{f}, \frac{a}{b} = c_{\vec{\alpha}} \text{ and is irreducible}}$,
with the convention that $\max \emptyset = 0$.

\subsection{Circuits for polynomials}
An \emph{(arithmetic) circuit} over $\mathbb{Q}[\vec{x}]$ is a labeled
directed acyclic graph (DAG) $\mathcal{C} = (V, E, \ell, o)$. Vertices with
indegree (a.k.a. \emph{fanin}) $0$ are called \emph{input gates} and they are
labeled by $\ell$ with some $x_i$ or a constant from $\Set{-1,0,1}$. The other
vertices all have fanin $2$, are called \emph{computation gates}, and are
labeled by one of the operations from the ring, i.e. $+$ and $\times$. The \emph{output gate} is $o \in V$. The
\emph{size} $|\mathcal{C}|$ of the circuit is $|V|$; its \emph{depth}, noted
$\depth{\mathcal{C}}$, is the length of a maximal path in $\mathcal{C}$. Every
vertex $v$ of the circuit \emph{computes} a polynomial function defined
inductively: If $v$ is an input gate, its function is the expression with
which it is labeled, i.e. $\sem{v} = \ell(v)$; otherwise, $\sem{v}$ is the
application of the operation $\ell(v)$ to all $\sem{u_i}$
such that $(u_i, v) \in E$. Finally, we write $\sem{\mathcal{C}} = \sem{o}$
for the function computed by the output gate.

Let $f \in \mathbb{Q}[\vec{x}]$ be a polynomial. If $f$ is given in its
sum-of-products representation, one can construct a circuit $\mathcal{C}$ of
size $O(|\supp{f}|\deg{f}\log(\coeff{f}))$ such that $\sem{\mathcal{C}}$ computes $f$.
Further, if the degree of $f$ is small (e.g. the exponent vectors are given in
unary), we can ensure $\mathcal{C}$ has poly-logarithmic depth.
\begin{theorem}[From~{\cite[Theorem
  2]{DBLP:journals/siamcomp/ValiantSBR83}}]\label{thm:depth-redux}
  Let $\mathcal{C}$ be a circuit over $\mathbb{Q}[\vec{x}]$ such that $d = \deg{\sem{\mathcal{C}}}$. Then, 
  we can construct a circuit $\mathcal{C}'$ of size
  $O((|\mathcal{C}|d^2)^3)$ and depth $O(\log(|\mathcal{C}|d) \log(d))$ that computes
  the same polynomial.
\end{theorem}

\subsection{Circuits for rational functions}
Arithmetic circuits over $\mathbb{Q}(\vec{x})$ have also been studied
in the literature. Compared to the ones defined above, we would need
to further allow for \emph{division gates}. In general, such circuits
compute rational functions instead of polynomials. A well-known fact
is that one can remove all division gates and replace them by a single
division at the output of the circuit: We carry through the gates both
the numerator and the denominator of the current function.
Multiplication affects both of them the same, but addition results
in three extra multiplication gates since we need a common denominator.
This pushing of divisions to the top of the circuit is attributed to
Strassen~\cite{strassen73}. The main result of that same paper is to
argue that if a circuit over $\mathbb{Q}(\vec{x})$ computes a polynomial, then
a polynomial-sized circuit over $\mathbb{Q}[\vec{x}]$ can compute the same function.

\begin{theorem}[From~\cite{strassen73}]\label{thm:no-div}
    Let $\mathcal{C}$ be a circuit over $\mathbb{Q}(\vec{x})$ computing a
    polynomial such that $d = \deg{\sem{\mathcal{C}}}$. Then, we can construct a circuit $\mathcal{C}'$ over $\mathbb{Q}[\vec{x}]$ of
    size $(|\mathcal{C}|d)^{O(1)}$ that computes the same polynomial.
\end{theorem}

\subsection{Parametric Markov chains}
A \emph{parametric Markov chain} (pMC, for short) is a tuple $\mathcal{M} =
(S,P)$ where $S$ is a finite set of states and $P \colon S \times S \to
\mathbb{Q}[\vec{x}]$ is a (parameterized) probabilistic transition function.
For convenience, we write $p_{ij}$ for $P(i,j)$, with the convention $S =
\Set{1, 2, \dots, n}$. A valuation $\vec{v} \in \mathbb{R}^m$ is \emph{graph preserving}
if for all $0 \leq i,j \leq n$ we have
\(
    p_{ij} \not\equiv 0 \implies 0 < p_{ij}(\vec{v}) \leq 1
\) and \(
    \sum_{k=1}^n p_{ik}(\vec{v}) = 1.
\) Intuitively, a graph-preserving valuation
allows us to instantiate a classical Markov chain (i.e. a nonparametric one) from $\mathcal{M}$. If the set of all graph-preserving valuations is  exactly $(0,1)^m$, we say it is a \emph{simple pMC}\footnote{Our definition of a simple pMC is more general than the one used in the literature (e.g., \cite{monotonicity,DBLP:journals/jcss/JungesK0W21})}.

For classical Markov chains, several interesting quantities like the
probability of eventually reaching a given state, are computable, e.g. via
linear programming~\cite{DBLP:books/wi/Puterman94}. When the Markov chain has
parameters, instead of a concrete probability, we can compute a rational
function $f$ in $\vec{x}$ which yields the corresponding value $f(\vec{v})$
given a graph-preserving valuation $\vec{v}$. In this work, we focus on
computing such functions for reachability.

\paragraph{Encoding and size of a pMC.} Let $f \in \mathbb{Q}[\vec{x}]$ be a polynomial. Henceforth, we write $\reps{f}$ for its representation size: $|\supp{f}|\deg{f}\log(\coeff{f})$. Intuitively, we
suppose the polynomial is represented as a sum of products with constants and
coefficients encoded as numerator-denominator pairs in binary while the exponent vectors are
encoded in unary. 
Then, by $|\mathcal{M}|$, we denote the \emph{size of the pMC},
that is $\sum_{i,j} \reps{p_{ij}}$. 

\section{Reachability value functions}
For this section, we fix a pMC $\mathcal{M} = (S,P)$ with $S =
\Set{1,\dots,n}$ and $p_{ij} \in \mathbb{Q}[\vec{x}]$ for all $1 \leq i,j \leq
n$. Without loss of generality, we assume that $n$ is the unique target state and that $n-1$ is the unique state that does not eventually reach $n$ with positive probability.
This is no
loss of generality since having such extremal reachability probabilities are
\emph{qualitative properties} that can be decided in polynomial time based on
the graph structure of the pMC only (see, e.g.,~\cite[Thm. 10.29 and Cor.
10.31]{DBLP:journals/iandc/BaierHHJKK20}) and all states that have the same extremal value can be ``merged'' while preserving reachability probabilities.  

Consider rational functions $\vec{g} = (g_1,\dots,g_n)$. We call these
\emph{(reachability) value functions} if they are a solution to the following
system of linear equations with polynomial coefficients. Intuitively, the $i^\text{th}$ value function $g_i$, for $1\leq i\leq n$, represents the probability of reaching state $n$ from state $i$. Below, $A$ is an $n
\times n$ matrix whose entries $A_{ij}$ are $p_{ij}$, for all $1 \leq i < n-1$
and all $1 \leq j \leq n$, and $0$ for $i \in \Set{n-1,n}$ and all $1 \leq j \leq n$.
\begin{equation}\label{eqn:reach}
    (\mathbb{I}_n - A)\vec{y} = \begin{bmatrix}
        \vec{0}\\
        1
    \end{bmatrix}
\end{equation}
The matrix $\mathbb{I}_n - A$ is guaranteed to be nonsingular for all graph-preserving
valuations (see, e.g.,~\cite[Theorem 10.19]{DBLP:books/daglib/0020348}) so
value functions exist and are even unique, up to extensional equivalence~\cite{DBLP:conf/ictac/Daws04}. Furthermore, they are known to be computable, e.g., using Gaussian
elimination~\cite{DBLP:journals/fac/LanotteMT07,DBLP:journals/iandc/BaierHHJKK20}.

\subsection{Monotonicity and being never worse}
We will now formally define two properties of states and their value functions. Namely, when the functions are \emph{monotonic} and when the value of a state is \emph{never worse} than that of another given state. We will continue with our fixed pMC $\mathcal{M} = (S,P)$ having $S = \Set{1,\dots,n}$ and $\vec{g} = (g_1,\dots,g_n)$ its value functions.

\paragraph{Monotonicity.} Let $i \in S$ be a state and $x_j$, for $1\leq j\leq m$, a parameter. We say that $\mathcal{M}$ is \textit{monotonically increasing}~\cite{monotonicity} from $i$ and in $x_j$, written $\moninc{\mathcal{M}}{x_j}{i}$, if $g_i(\vec{u})\leq g_i(\vec{v})$ for all graph-preserving valuations $\vec{u},\vec{v}$ such that $\vec{v} - \vec{u} = \varepsilon \vec{e_j}$ for some $\varepsilon \in\mathbb{R}_{\geq0}$, where $\vec{e_j}$ is the null vector with $1$ only in the $j^\text{th}$ position. Being \textit{monotonically decreasing} can be defined analogously.

\paragraph{The Never-Worse Relation.}
Let $i, j \in S$ be states. We say that $j$ is \textit{never-worse} than $i$, denoted $i\trianglelefteq j$, if $g_i(\vec{v})\leq g_j(\vec{v})$ for all graph-preserving valuations $\vec{v}$. The relation $\trianglelefteq$ is called the \emph{never-worse relation}~\cite{DBLP:conf/atva/EngelenPR23,DBLP:conf/fossacs/0001P18}, or NWR for short.

The natural decision problems associated with monotonicity and the NWR are coETR complete \cite{monotonicity,DBLP:conf/atva/EngelenPR23}. For monotonicity, the problem asks, given a pMC $\mathcal{M}$, a state $i$, and a parameter $x_j$, whether $\moninc{\mathcal{M}}{x_j}{i}$ holds. For the NWR, the problem asks, given a pMC $\mathcal{M}$, and two states $i$ and $j$, whether $i \trianglelefteq j$ holds. A clear relationship between the two problems has been established in one direction.

\begin{theorem}[From~{\cite[Lemma 2]{monotonicity}}]\label{thm:nwr2mon}
    There is a polynomial-time many-one reduction from the NWR problem to the problem.
\end{theorem}
\begin{proof}[Proof sketch]
    Add a new state $n+1$, a new parameter $x_{m+1}$, and two new transitions to the states for which we want to check the NWR: with probability $x_{m+1}$ we transition from $n+1$ to $j$; with $1-x_{m+1}$, from $n+1$ to $i$. The new pMC is monotonically increasing from $n+1$ in $x_{m+1}$ if and only if $i\trianglelefteq j$ in the original pMC.
\end{proof}

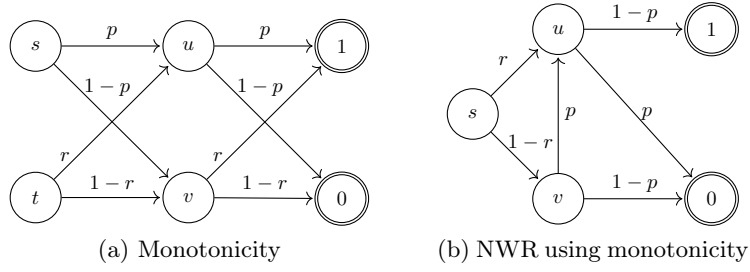
\begin{figure}
    \centering
\subfloat[Monotonicity\label{fig:just_mono}] {
    \begin{tikzpicture}[shorten >=1pt,auto,node distance=1.9 cm, scale = 0.7, transform shape]

        \node[state](s){$s$};
        \node[state](t)[below=of s]{$t$};
        \node[state](u)[right=of s]{$u$};
        \node[state](v)[below=of u]{$v$};
        \node[state,accepting](fin)[right=of u]{$1$};
        \node[state,accepting](fail)[below=of fin]{$0$};
                
        \path[->,auto] 
          (s)   edge node {$p$} (u)
                edge node [right,pos=0.2] {$1-p$} (v)
          (t)   edge node [left,pos=0.2] {$r$} (u)
                edge node {$1-r$} (v)
          (u)   edge node {$p$} (fin)
                edge node [right,pos=0.2] {$1-p$} (fail)
          (v)   edge node [left,pos=0.2] {$r$} (fin)
                edge node {$1-r$} (fail)
        ;
    \end{tikzpicture}
}\qquad
\subfloat[NWR using monotonicity\label{fig:NWR_using_mono}] {
    \begin{tikzpicture}[shorten >=1pt,auto,node distance=1.9 cm, scale = 0.7, transform shape]

        \node[state](s){$s$};
        \node[state](u)[above right=1.3cm of s]{$u$};
        \node[state](v)[below right=1.3cm of s]{$v$};
        \node[state,accepting](fin)[right=of u]{$1$};
        \node[state,accepting](fail)[right=of v]{$0$};
                
        \path[->,auto] 
          (s)   edge node {$r$} (u)
                edge node [right,pos=0.2] {$1-r$} (v)
          (u)   edge node {$1-p$} (fin)
                edge node [right] {$p$} (fail)
          (v)   edge node [right] {$p$} (u)
                edge node {$1-p$} (fail)
        ;
    \end{tikzpicture}
}
    \caption{Examples of pMCs with parameters $p$ and $r$.}
    \label{fig:monoNWR}
\end{figure}

\begin{example}
    \Cref{fig:just_mono} shows an example of a (simple) pMC with $6$ states where $1$ is the target state. It is easy to verify that the value functions $g_s=p^2+r-rp$ and $g_t=rp+r-r^2$ are monotonically increasing in $r$ and $p$ (respectively), but not in $p$ and $r$ (respectively). In \Cref{fig:NWR_using_mono}, we have $v\trianglelefteq u$ since $g_v=p-p^2<1-p=g_u$ for all graph preserving valuations. This relation can be witnessed from the state $s$ by checking monotonicity in $r$. In this case, $g_s=r(1-p)^2+p(1-p)$ is indeed monotonically increasing in $r$.
\end{example}

The reduction from \Cref{thm:nwr2mon} is useful because it allows any (approximation) algorithm developed for the monotonicity problem to be translated into one for the NWR. We would like to establish a relationship in the opposite direction. Towards that, we present the following reduction which assumes the given pMC $\mathcal{M}$ is simple and that it is paired with a \emph{partial-derivative pMC}.

\begin{theorem}\label{thm:conditional}
    Consider being given simple pMCs $\mathcal{M} = (S,P),\mathcal{M}'$ over $\vec{x}$ with value functions $\vec{g},\vec{g'}$, $i \in S$, and $0 \leq k \leq m$, such that $\frac{\partial g_i}{\partial x_k} = \beta + N g'_1$ where $N \in \mathbb{N}_{>0}, \beta \in \mathbb{Q}$. Then, there is a polynomial-time (in the size of $\mathcal{M}'$) many-one reduction from the monotonicity problem for $\mathcal{M}$ to the NWR problem.
\end{theorem}
\begin{proof}
    Since $\mathcal{M}$ is simple, its set of graph-preserving valuations is connected (in the topological sense). Therefore,  $\moninc{\mathcal{M}}{i}{x_k}$ holds if and only if $\frac{\partial g_i}{\partial x_k}(\vec{v}) \geq 0$ for all graph-preserving valuations $\vec{v}$ (cf.~\cite[Def. 1]{monotonicity} and~\cite[Def. 2]{DBLP:conf/tacas/SpelJK21}). Because $\mathcal{M}'$ is also simple, the sets of graph-preserving valuations of the two pMCs coincide. Hence, to check for monotonicity in $\mathcal{M}$ we focus on the value function $g'_1$ of $\mathcal{M}'$. 
    
    Now, if $\frac{\beta}{N} \geq 0$ we will have monotonicity and if $\frac{\beta}{N} \leq -1$ we can conclude the opposite (unless $g'_1 = 1$). Since all of the latter can be checked in polynomial time, we suppose $-1 < \frac{\beta}{N} < 0$, which is the same as $0 < \frac{-\beta}{N} < 1$. In this case, monotonicity holds if and only if $g'_1(\vec{v}) \geq \frac{-\beta}{N}$ for all graph-preserving valuations $\vec{v} \in (0,1)^m$. This concludes the reduction because that property is straightforward to encode as an instance of the NWR problem: we can add a new state $s$ to $\mathcal{M}'$ that transitions with probability $\frac{-\beta}{N}$ to the target and $1+\frac{\beta}{N}$ to the sink, and then ask whether $s \trianglelefteq 1$.
\end{proof}
We observe that the proof of the theorem does not crucially rely on $\mathcal{M}$ being a simple pMC. The same argument can be used for any pMC whose graph-preserving valuations form a connected subset of $(0,1)^m$ if it is definable via a system of strict polynomial inequalities $\vec{0} < \vec{A} < \vec{1}$. Then, the entries of $\vec{A}$ could be put as transition probabilities between new ``dummy'' states in $\mathcal{M}'$ to constrain its set of graph-preserving valuations.

Unfortunately, we do not know whether we incur in a loss of generality because of the assumption that the partial-derivative pMC is given. In the sequel, we do provide a quasi-polynomial construction for a simple partial-derivative pMC. Together with \Cref{thm:conditional}, this gives us a quasi-polynomial reduction from the monotonicity problem to the NWR problem for simple pMCs.

\section{Computing representations of value functions}
In this section, we recall the main ideas behind how one can compute the value functions of a given pMC. We argue that a suitable polynomial-size representation of them via circuits can be efficiently computed by symbolically executing the known algorithms. Then, we use this and a connection between circuits and yet another representation of polynomial functions to give a quasi-polynomial construction of partial-derivative pMCs.

\subsection{Value functions via one-step division-free Gaussian elimination}
Using one-step division-free Gaussian elimination~\cite{bareiss}, as proposed
in~\cite[Section 3.2]{DBLP:journals/iandc/BaierHHJKK20}, \Cref{eqn:reach} can
be put into \emph{row echelon form} while guaranteeing all entries are still
polynomials. In fact, we obtain the following system
of linear equations with a diagonal matrix of polynomial coefficients that is equivalent to
\Cref{eqn:reach} (see~\cite[Algorithm 1]{DBLP:journals/iandc/BaierHHJKK20}).
\begin{equation}\label{eqn:row-echelon}
  \begin{bmatrix}
    a_{11} & \dots & 0\\
    \vdots & \ddots & \vdots\\
    0 & \dots & a_{nn}
  \end{bmatrix}
  \vec{y} = \begin{bmatrix}
    b_1\\
    \vdots\\
    b_n
  \end{bmatrix}
\end{equation}
Furthermore, the $a_{ii}$ and $b_{i}$ are all polynomials and their degree is at most
$O(nd)$, where $d = \max_{i,j} \deg{p_{ij}}$ (see~\cite[Proof of Lem.
3]{DBLP:journals/iandc/BaierHHJKK20}). The $a_{ii}$ and $b_{i}$ are
obtained from the $p_{ij}$ polynomials by applying addition, multiplication,
subtraction, and division --- $O(n^3)$ times, to be precise. Hence, for each $a_{ii}$ or $b_{i}$ we can
construct a circuit $\mathcal{C}$ over
$\mathbb{Q}(\vec{x})$ of polynomial size, in $n$ and $d$, such that the polynomial is computed by $\mathcal{C}$. Since $b_{i} / a_{ii}$, for all $1 \leq i \leq n$, is a solution to the system, appealing to \Cref{thm:no-div}, followed by \Cref{thm:depth-redux}, we obtain the following.
\begin{theorem}\label{thm:circuit-val}
    Let $\mathcal{M} = (S,P)$ be a pMC with $S = \Set{1,\dots,n}$ and write $d = \max_{i,j} \deg{p_{ij}}$. Then, we can construct two families of circuits $\mathcal{N}_i, \mathcal{D}_i$ over $\mathbb{Q}[\vec{x}]$ of size $|\mathcal{M}|^{O(1)}$ and depth $O(\log(|\mathcal{M}|) \log(nd))$ such that $\sem{\mathcal{N}_i}/\sem{\mathcal{D}_i}$ is the $i$-th value function $g_i$, for all $1 \leq i \leq n$.
\end{theorem}
This means that, while the sum-of-products representation of value functions may be exponential~\cite[Lemma 2]{DBLP:journals/iandc/BaierHHJKK20}, their circuit representation is always small!

Before we close this subsection, we leverage the previous theorem to establish
that also the partial derivatives of the value functions always have a small
circuit representation. For this, we abuse our notation for circuits to allow
for multiple (ordered) output gates so that a circuit can compute a vector-valued function. We need the following result about
partial derivatives of functions as circuits.
\begin{theorem}[From~\cite{DBLP:journals/tcs/BaurS83,DBLP:journals/sigact/Morgenstern85}]\label{thm:part-der}
    Let $\mathcal{C}$ be a circuit over $\mathbb{Q}(\vec{x})$ that computes a (scalar) function $f$. Then, we can construct a circuit $\mathcal{C}'$ over $\mathbb{Q}(\vec{x})$ of size $O(|\mathcal{C}|)$ and depth $O(\depth{\mathcal{C}})$ that  computes $(f, \frac{\partial f}{\partial x_1},\dots,\frac{\partial f}{\partial x_m})$.
\end{theorem}
Now, we can use the construction from \Cref{thm:part-der} on the circuits representing the numerator and denominator of the value functions in \Cref{thm:circuit-val}. Then, using the quotient rule, we obtain the following.
\begin{theorem}\label{thm:part-der-val}
    Let $\mathcal{M} = (S,P)$ be a pMC with $S = \Set{1,\dots,n}$ and write $d = \max_{ij} \deg{p_{ij}}$. Then, we can construct two families of circuits $\mathcal{N}_i,\mathcal{D}_i$ over $\mathbb{Q}[\vec{x}]$ of size $|\mathcal{M}|^{O(1)}$ and depth $O(\log(|\mathcal{M}|)\log(nd))$ such that $\sem{\mathcal{N}_i} / \sem{\mathcal{D}_i}$ gives $g_i$ and its partial derivatives $(g_i, \frac{\partial g_i}{\partial x_1},\dots,\frac{\partial g_i}{\partial x_m})$, for all $1 \leq i \leq n$.
\end{theorem}

\subsection{Quasi-polynomial partial-derivative chains}
In this subsection we aim at an extension of~\Cref{thm:part-der-val} where we get, for a value function $g_i$ and a given variable $x_j$,
a pMC such that one of its value functions corresponds to $\frac{\partial g_i}{\partial x_j}$ (cf.~\cite[Section 3.2]{DBLP:conf/vmcai/HeckSJMK22}, where the authors construct parametric weighted automata computing the partial derivatives). 
To realize this, we use a classical result due to Ben-Or and Cleve~\cite{DBLP:journals/siamcomp/Ben-OrC92} to go from the circuit representation of the partial derivative to something closer to pMCs. Namely, we obtain an \emph{algebraic branching program}.

An algebraic branching program (ABP, for short)~\cite{DBLP:conf/stoc/Nisan91} over $\mathbb{Q}[\vec{x}]$ is a labeled DAG $\mathcal{A} = (V,E,\ell)$ with $d+1$ layers of vertices $V$. This means that $V$ is partitioned into $V_0, \dots, V_d$ such that $(u,v) \in E$ implies $u \in V_i$ and $v \in V_{i+1}$ for some $0 \leq i < d$. The \emph{size} of $\mathcal{A}$ is $|V|$; its \emph{width} is the cardinality of the largest layer, $\max_{i=0}^d |V_i|$. We will be interested in the single-source single-sink kind of ABP, meaning the first and last layers have a single vertex. Every edge $e \in E$ is labeled by a linear term $\ell(e)$ over $\vec{x}$, i.e. $\ell: E \to \mathbb{Q}[\vec{x}]$ and $\deg{\ell(e)} \leq 1$ for all $e \in E$. The linear terms are encoded just like polynomials are in pMCs (as sum of products with binary encoding for coefficients and constants, unary for exponents).
Write $P(\mathcal{A})$ for the set of all maximal paths in $\mathcal{A}$. We say that $\mathcal{A}$ computes the polynomial $\sem{\mathcal{A}} = \sum_{\pi \in P(\mathcal{A})} \prod_{e \in \pi} \ell(e)$. Just like in circuits, each vertex $v$ computes a polynomial $\sem{v}$, i.e. taking them as the single source of the ABP.
The result below links circuits and ABPs in the required direction.
\begin{theorem}[From~{\cite[Theorem 1]{DBLP:journals/siamcomp/Ben-OrC92}}]\label{thm:ben-or-cleve}
    Let $\mathcal{C}$ be a circuit over $\mathbb{Q}[\vec{x}]$. Then, we can construct an ABP $\mathcal{A}$ of size $O(4^{\depth{\mathcal{C}}})$ and width $4$ that computes the same polynomial.
\end{theorem}
Using the result above on \Cref{thm:part-der-val}, we get an ABP as stated below.
\begin{theorem}\label{thm:pmc2pds}
Let $\mathcal{M} = (S,P)$ be a pMC with $S = \Set{1,\dots,n}$ and write $d = \max_{ij} \deg{p_{ij}}$. Then, we can construct two families of ABPs $\mathcal{N}_i,\mathcal{D}_i$ over $\mathbb{Q}[\vec{x}]$ of size $|\mathcal{M}|^{O(\log(nd))}$ and width $4$ such that $\sem{\mathcal{N}_i} / \sem{\mathcal{D}_i}$ gives the partial derivatives $(\frac{\partial g_i}{\partial x_1},\dots,\frac{\partial g_i}{\partial x_m})$ of $g_i$, for all $1 \leq i \leq n$.
\end{theorem}

Next, we will give a translation from ABPs to pMCs. The following result will be useful to rewrite polynomials into a form that is easy to encode as probability distributions. For completeness, and because it is a good warm-up for the sequel, we also give a proof of the result.
\begin{lemma}[Chonev's trick~{\cite[Remark 1]{DBLP:conf/rp/Chonev19}}]\label{lem:chonev}
    Let $f \in \mathbb{Q}[\vec{x}]$ be a polynomial. Then, there are $c \in \mathbb{Z}$, $d, N \in \mathbb{N}_{>0}$, and $\vec{a}, \vec{b}
    \in \mathbb{N}^T_{> 0}$ with $\frac{|c|}{d}+\sum_{i=1}^T \frac{a_i}{b_i} \leq 1$, such that $f = N(\frac{c}{d} + 
    \sum_{i=1}^T \frac{a_i}{b_i} Q_i(\vec{x}))$ where the $Q_i$ are 
    products of at most $\deg{f}$ terms from $\bigcup_{i=1}^m \Set{x_i, 1- x_i}$. Moreover,
    \begin{itemize}
        \item $T \leq |\supp{f}|\deg{f}$ and
        \item $\log N, \log |c|, \log d, \log a_i, \log b_i$ are all $O(\reps{f})$.
    \end{itemize}
\end{lemma}
\begin{proof}
    Starting from a sum-of-products representation of $f$, we choose some monomial $t_{\vec{\alpha}} \vec{x}^{\vec{\alpha}}$ such that $t_{\vec{\alpha}}$ is negative and $\vec{\alpha} \neq \vec{0}$, i.e. it is not the constant term. Now, we choose one of the variables $x_i$ and replace it with $1 - x_i$ in the product so that we can flip the sign of $t_{\vec{\alpha}}$. To obtain a sum of terms equivalent to the original monomial, we need to cancel out the new monomial we have introduced. The latter has degree strictly smaller than the original monomial so we can repeat our approach until we get a constant (whose sign does not matter in the claim). The original example by Chonev is quite instructive: If we start with $-2x_1x_2x_3$, we can rewrite it as:
    \(
    2(1-x_1)x_2x_3 + 2(1-x_2)x_3 + 2(1-x_3) - 2.
    \)

    Observe that, in the rewriting process described above, every monomial is
    transformed into a sum of at most $\deg{f}$ terms. The first bound in the
    claim follows from this observation. For the second bound, we note that
    all coefficients $\frac{a_i}{b_i}$ and the constant $\frac{c}{d}$ are sums of rational constants from $f$
    whose bitsize is at most $\log(\coeff{f})$. Now we consider, in turn, each of the operations which make the bitsize of the result larger.
    \begin{description}
        \item[To obtain a common denominator] for the sums, in the worst case, we need to multiply all
        denominators. This results in adding $\log(\coeff{f})$ to the bitsize (of both the numerator and denominator) at most $T$ times.
        \item[The sums to obtain
        each numerator] $a_i$ may incur in a further blow-up of the bitsize by adding
        $T$ to it. So far, the bitsize of $|c|,d,a_i,b_i$ is at most:
        \begin{align*}
            & (|\supp{f}|\deg{f} + 1)\log(\coeff{f}) + |\supp{f}|\deg{f} \\
            {}\leq{} & (|\supp{f}|\deg{f} + 1)(\log(\coeff{f})+1).
        \end{align*}
        \item[To get a sum of at most one] and to factor out $N$, we can set $N$ to be the sum of all numerators $a_i$. Then, all denominators get multiplied by $N$ and we can factor it out as required. Hence, the bitsize of $N$ is at most:
        \begin{align*}
            & (|\supp{f}|\deg{f} + 1)(\log(\coeff{f}) + 1) + |\supp{f}|\deg{f} \\
            {}\leq{} & (|\supp{f}|\deg{f} + 1)(\log(\coeff{f}) + 2),
        \end{align*}
        while the final bitsize of the $b_i$ is at most double that of $N$ due to their multiplication by $N$.
    \end{description} 
    The claim thus follows from the definition of $\reps{f}$.
\end{proof}

Now, using Chonev's trick, we can inductively construct an acyclic pMC from any ABP such that one of the value functions of the former and the polynomial computed by the latter have a linear relation.
\begin{lemma}\label{lem:ind-constr}
    Let $\mathcal{A} = (V,E,\ell)$ be an ABP over $\mathbb{Q}[\vec{x}]$ of width $w$ and write $s = \max_{e \in E} |\supp{\ell(e)}|$. Then, we can compute $N \in \mathbb{N}_{>0}, \beta \in \mathbb{Q}$ and a simple pMC $\mathcal{M} = (S,P)$ with value functions $\vec{g}$ such that $\sem{\mathcal{A}} = \beta + Ng_1$, $|S|$ is $O(|V|sw)$ and $\beta$, $N$, and all $p_{ij}$, have polynomial bitsize in $|V|,w,\max_{e \in E} \reps{\ell(e)}$.
\end{lemma}
Before we prove the lemma, we observe that together with \Cref{thm:pmc2pds} it implies that we can construct quasi-polynomial partial-derivative pMCs.
\begin{theorem}\label{thm:quasi}
    Let $\mathcal{M} = (S,P)$ be a pMC with value functions $\vec{g}$ and $d = \max_{i,j} \deg{p_{ij}}$. Then, for any $i \in S$ and $0 \leq k \leq m$, we can compute $N \in \mathbb{N}_{>0},\beta \in \mathbb{Q}$ and a simple pMC $\mathcal{M}'$ of size $|\mathcal{M}|^{O(\log(|S|d))}$ with value functions $\vec{g'}$ s.t. $\frac{\partial g_i}{\partial x_k} = \beta + N g'_1$ where $\beta$ and $N$ have $|\mathcal{M}|^{O(\log(|S|d))}$ bitsize representations.
\end{theorem}

\begin{proof}[of \Cref{lem:ind-constr}]
  Let $\mathcal{A}$ have $L + 1$ layers $V_0,\dots,V_L \subseteq V$. Then, for
  all $u \in V_{L-1}$ we have that $\sem{u} = \ell(u,t)$, where $t \in V_L$ is
  the unique sink. Similarly, for all $0 \leq j < L -1$ and all vertices $u \in
  V_j$:
  \begin{equation}\label{eqn:chonev-rec}
    \sem{u} = \sum_{v \in V_{j+1}} \ell(u, v) \sem{v}.
  \end{equation}
  Based on the equation above, we will give an inductive construction of a pMC
  $\mathcal{M}$ based on Chonev's trick now. 

  Our induction is on the layer $j$ and we will argue that for all $u \in
  V_{j}$ we have states $\overline{u},\underline{u}$ in the pMC such that
  $\sem{u}/N_u - \beta_u = g_{\overline{u}} = 1 - g_{\underline{u}}$.
  First, we add states $\bot$ and $\top$ to $\mathcal{M}$ with the intention
  of $\top$ being the target state and $\bot$ a state with only itself as
  successor. For $j = L - 1$ and each vertex $u \in V_j$ we add to
  $\mathcal{M}$ states $\overline{u},\underline{u}$.  Using Chonev's trick, we
  then rewrite $\ell(u,t)$ as $N_u(\beta_u + \sum_{k=1}^T \alpha_k
  Q_k(\vec{x}))$. Note that since $\ell(u,t)$ is a linear function, $T \leq 2s$
  and the $Q_k(\vec{x})$ are also linear. For each $Q_k(\vec{x})$ we add a
  state to $\mathcal{M}$ and a transition from it to $\top$ with probability
  $Q_k(\vec{x})$ and one to $\bot$ with probability $1 - Q_k(\vec{x})$. From
  $\overline{u}$ we add transitions to the states $Q_k$ states with
  corresponding probability $\alpha_k$ and one to $\bot$ with probability $1 -
  \sum_{k=1}^T \alpha_k$. To conclude, we add another state for each $Q_k$ and
  transitions from them to $\bot$, instead of $\top$, with probability $Q_k$
  and transitions to $\top$ with $1-Q_k$, as well transitions from
  $\underline{u}$ to these new copies of $Q_k$ with the $\alpha_k$ as
  probabilities and to $\top$, not $\bot$, with $1 - \sum_{k=1}^T \alpha_k$.
  It is easy to check all the desired properties hold so the claim holds
  for some $j$.

  For the inductive step, consider $0 \leq j < L - 1$ and
  a vertex $u \in V_j$. We again add states $\overline{u},\underline{u}$ to $\mathcal{M}$ and
  consider \Cref{eqn:chonev-rec}. From the inductive hypothesis, 
  \begin{equation}\label{eqn:expand}
    \sem{u} = \sum_{v \in V_{j+1}} \ell(u, v) N_v\left(\beta_{v} +
    g_{\overline{v}}\right) = N_v\sum_{v \in V_{j+1}} \beta_{v} \ell(u,v) + 
    \ell(u,v) g_{\overline{v}}.
  \end{equation}
  Interpreting the $g_{\overline{v}}$ as variables for a moment, we can again
  use Chonev's trick to rewrite the above as $N_u(\beta_u + \sum_{k=1}^T
  \alpha_k Q_k(\vec{x}))$. This time we started from a sum of (at most $sw$)
  quadratic terms, so $T \leq 3sw$ and the $Q_k(\vec{x})$ are at most
  quadratic. For each $Q_i$ we construct a chain of length at most $2$ leading
  to $\top$, $\overline{v}$, or $\underline{v}$ depending on whether it has no
  term with $g_{\overline{v}}$ as a factor, it has $g_{\overline{v}}$ as a
  factor, or it has $(1 - g_{\overline{v}})$ as a factor. From $\overline{u}$
  we add transitions to the start of the new chains with corresponding
  probabilities $\alpha_k$ and to $\bot$ with the remaining probability.
  Similarly, we construct chains of length at most
  $2$ leading to $\bot$, $\underline{v}$, or $\overline{v}$ with the same
  conditions as before, in the same order. Then, from $\underline{v}$ we add
  transitions to the new chains and to $\top$ with the remaining probability.
  Once more the desired properties hold and we thus conclude the description
  of how to construct the simple pMC.

  The bound on the number of states from the pMC is immediate from the construction and the
  bounds on the number of chains constructed per state per layer (all linear
  in $sw$) and their size being constant. For the bounds on the bitsize of the
  numbers, we revisit the cases considered for the proof of~\Cref{lem:chonev} with special attention to the substitution applied to get \Cref{eqn:expand}. Write $M$ for the value $\max_{e \in E} \log(\coeff{\ell(e)})$. 
  \begin{description}
        \item[To obtain a common denominator] for the sums, we may need to multiply all
        denominators. However, we can make sure that all $\beta_v$ (from the same layer) have the same common denominator. This means that it suffices to compute a common denominator for the coefficients of the $\ell(u,v)$ and multiply it by that of the $\beta_v$.        
        Hence, in each layer, we add $M$ to the bitsize (of both the numerator and denominator) at most $T \leq 3sw$ times.
        \item[The sums to obtain
        each numerator] may incur in a further blow-up of the bitsize by adding
        $T \leq 3sw$ to it. This is on top of the $M$ additional bits relative to the representation of the numerator of the $\beta_v$.
        \item[To get a sum of at most one] and to factor out $N_u$ we again need to be careful. Note that we can make sure all the $N_v$ (from the same layer) are equal. Now, since $\beta_v \leq 1$, it suffices to add the absolute values of all numerators of coefficients and constants from the $\ell(u,v)$ (for all $v$ in the same layer), call that $A$, and to set $N_u = A N_v$. Then, all denominators get multiplied by $A$ too and we get an increase in bitsize, for numerators and denominators and for $N_u$ relative to $N_v$, of at most $3swM$.
    \end{description}
    From the above analysis we get that every layer results in adding $O(swM)$ to the bitsize of the integers required to write down the polynomials we manipulate. Since the number of layers is at most $|V|$, 
    all the $\reps{p_{ij}}$ are bounded by $O(|V|w\max_{e \in E} \reps{\ell(e)})$ as required.
\end{proof}

Unfortunately, it seems difficult to improve \Cref{thm:quasi} and get a polynomial-time reduction using techniques from arithmetic circuits and branching programs as we have. Indeed, the question of whether the former can be translated into polynomial-sized branching programs (a.k.a. \emph{skewed circuits}) seems to be open (cf.~\cite[Sec. 3, p. 10]{DBLP:journals/corr/Mahajan13}). It is also not clear how to adapt our inductive translation from ABPs to get pMCs from the weighted automata used in~\cite[Section 3.2]{DBLP:conf/vmcai/HeckSJMK22} to represent partial derivatives of value functions. Here, the main complication is that the automata may have cycles.

\section{NWR equivalence classes}
In the rest of the paper, we adopt a practical approach in reducing the state space of pMCs by collapsing NWR equivalent states. To keep the algorithm efficient and to avoid the difficulties of ETR-hardness, we focuss on a subclass of pMCs, which we define below. It is important to note that even in a general pMC, if we ignore the parameters to obtain a ``trivial'' pMC, any equivalence in the new pMC is also one in the original pMC.

A \textit{trivially parametric} Markov chain is a pMC where the polynomial on each transition $t$ is the linear function $x_t$, so $x_t$ appears uniquely in the label of $t$. It is known that equivalence classes of trivially parametric Markov chains can be computed in polynomial time \cite{DBLP:conf/atva/EngelenPR23}. This is due to the fact that every equivalence class has a unique ``exit'', that is, a unique state that has transitions leaving the equivalence class. This exit is the closest state to the target among all states in the equivalence class. It follows that equivalence classes can be found by doing a reverse breadth-first search (BFS) from the target state to find these exits and then looking for states that almost-surely reach these exits (see \Cref{algo:eq}).

One can easily verify that the algorithm runs in $O(n^2)$ time, where $n$ is the number of states of the pMC. The extremal states can certainly be computed and collapsed in $O(n^2)$~\cite{DBLP:books/daglib/0020348}. Furthermore, each state $u\in S$ is visited at most once during the reverse BFS from the final state and computing the set of all states which have a path to the final or fail states after removing $u$ takes linear time. 

\begin{algorithm}[hbt]
\caption{Compute and collapse equivalence classes}\label{algo:eq}
\Input{A pMC $\mathcal{M}=(S,P)$, final state $\good$ and fail state $\bad$}
\Output{The pMC with all equivalence classes collapsed}
    Contract extremal-value states\;
    $\text{TODO} \gets S$\;
    $\text{EC} \gets \emptyset$\;
    Order the states in reverse BFS order, starting from the final state\;
    \ForEach{$u \in S$}{
        \If{$u \not \in \text{TODO}$} {
            \Continue\;
        }
        Determine set $U$ of states that have a path to $\good$ or a path to $\bad$ without going through the state $u$\;
        $\tilde{u} \gets \Set{u} \cup (V \setminus U)$\;
        $\text{TODO} \gets \text{TODO} \setminus \tilde{u}$\;
        $\text{EC} \gets \text{EC} \cup \Set{u}$\;
    }
    \ForEach{$u \in \text{EC}$}{
        Collapse all states in $\tilde{u}$ into $u$\;
    }
\end{algorithm}

The following is immediate from the definitions of NWR and monotonicity.
\begin{remark}
    Collapsing NWR equivalence classes preserves monotonicity.
\end{remark}

We implemented our algorithm and tested its performance against a number of benchmarks. These benchmarks can be divided into three categories: those from Spel et al. \cite{DBLP:conf/tacas/SpelJK21} which focus on the monotonicity analysis module of Storm; those from Heck et al. \cite{DBLP:conf/vmcai/HeckSJMK22}; and a number of benchmarks we constructed ourselves based on the properties of our algorithm. Our code and experimental data can be found on Zenodo \cite{paper_code}.

Our experiments aim to answer the following three questions:
\begin{itemize}
    \item[Q1.] How much does \Cref{algo:eq} reduce the size of (parametric) Markov chains?
    \item[Q2.] Does using the algorithm as a pre-processing step cause the monotonicity analysis algorithm from \cite{DBLP:conf/tacas/SpelJK21} to run more efficiently?
    \item[Q3.] Does our pre-processing algorithm make parameter lifting \cite{DBLP:conf/vmcai/HeckSJMK22} more efficient?
\end{itemize}

\subsection{Experimental set-up}

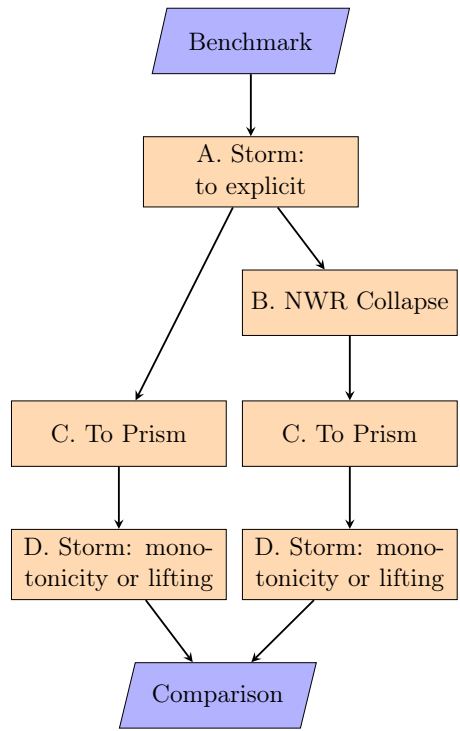
\begin{figure}[h]
    \centering
    \resizebox{0.4\textwidth}{!}{
    \begin{tikzpicture}[node distance=2cm]

    \node (benchmark) [io] {Benchmark};
    
    \node (to_explicit) [process, below of=benchmark] {A. Storm: to explicit};
    
    \node (nwr_collapse) [process, below of = to_explicit, xshift=1.5cm] {B. NWR Collapse};
    
    \node (to_prism_col) [process, below of=nwr_collapse] {C. To Prism};
    \node (to_prism_ref) [process, left of=to_prism_col, xshift=-1.5cm] {C. To Prism};

    \node (storm_col) [process, below of=to_prism_col] {D. Storm: monotonicity or lifting};
    \node (storm_ref) [process, below of=to_prism_ref] {D. Storm: monotonicity or lifting};
    
    \node (comparison) [io, below of=storm_ref, xshift=1.5cm] {Comparison};
    
    \draw [arrow] (benchmark) -- (to_explicit);
    \draw [arrow] (to_explicit) -- (nwr_collapse);
    \draw [arrow] (nwr_collapse) -- (to_prism_col);
    \draw [arrow] (to_explicit) -- (to_prism_ref);
    \draw [arrow] (to_prism_col) -- (storm_col);
    \draw [arrow] (to_prism_ref) -- (storm_ref);
    \draw [arrow] (storm_col) -- (comparison);
    \draw [arrow] (storm_ref) -- (comparison);
    
    \end{tikzpicture}
    }
    \caption{Overview of the experimental setup.}
    \label{fig:experimental_flow_chart}
\end{figure}

\begin{figure}[p]
    \centering
\subfloat[Variant A. \label{fig:reach0_ptime_var_a}] {
    \begin{tikzpicture}[shorten >=1pt,auto,node distance=.9 cm, scale = 0.5, transform shape,initial text={}]
    
    \node[state, initial](start){$s$};

    \node[state](np1i1l1)[above right=of start]{$n+1$};
    \node[state](1i1l1)[above=of np1i1l1]{$1$};
    \node[](d1i1l1)[right=.2cm of 1i1l1]{$\dots$};
    \node[](d2i1l1)[right=.2cm of np1i1l1]{$\dots$};
    \node (i1l1) [draw=red, fit= (np1i1l1) (d1i1l1) (1i1l1), inner sep=0.1cm, fill=red!20, fill opacity=0.2]{};
    \node[](d3i1l1)[below right=of d1i1l1]{$\dots$};

    \node[state](1i1l2)[below right=of start]{$1$};
    \node[state](np1i1l2)[below=of 1i1l2]{$n+1$};
    \node[](d1i1l2)[right=.2cm of 1i1l2]{$\dots$};
    \node[](d2i1l2)[right=.2cm of np1i1l2]{$\dots$};
    \node (i1l2) [draw=red, fit= (np1i1l2) (d1i1l2) (1i1l2), inner sep=0.1cm, fill=red!20, fill opacity=0.2]{};
    \node[](d3i1l2)[below right=of d1i1l2]{$\dots$};

    \node[state](1ikl1)[right=2cm of d1i1l1]{$1$};
    \node[state](2ikl1)[right=of 1ikl1]{$2$};
    \node[](d1ikl1)[right=.2cm of 2ikl1]{$\dots$};
    \node[state](nikl1)[right=.3cm of d1ikl1]{$n$};
    \node[state](np1ikl1)[below=of 1ikl1]{$n+1$};
    \node[state](np2ikl1)[below=of 2ikl1]{$n+2$};
    \node[](d2ikl1)[right=.2cm of np2ikl1]{$\dots$};
    \node[state](2nikl1)[right=.2cm of d2ikl1]{$2n$};
    \node (ikl1) [draw=red, fit= (np1ikl1) (2nikl1) (nikl1), inner sep=0.1cm, fill=red!20, fill opacity=0.2]{};

    \node[state](1ikl2)[right=2cm of d1i1l2]{$1$};
    \node[state](2ikl2)[right=of 1ikl2]{$2$};
    \node[](d1ikl2)[right=.2cm of 2ikl2]{$\dots$};
    \node[state](nikl2)[right=.3cm of d1ikl2]{$n$};
    \node[state](np1ikl2)[below=of 1ikl2]{$n+1$};
    \node[state](np2ikl2)[below=of 2ikl2]{$n+2$};
    \node[](d2ikl2)[right=.2cm of np2ikl2]{$\dots$};
    \node[state](2nikl2)[right=.2cm of d2ikl2]{$2n$};
    \node (ikl1) [draw=red, fit= (np1ikl2) (2nikl2) (nikl2), inner sep=0.1cm, fill=red!20, fill opacity=0.2]{};

    \node[state](1ikp1l1)[right=of nikl1]{$1$};
    \node[state](np1ikp1l1)[below=of 1ikp1l1]{$n+1$};
    \node[](d1ikp1l1)[right=.2cm of 1ikp1l1]{$\dots$};
    \node[](d2ikp1l1)[right=.2cm of np1ikp1l1]{$\dots$};
    \node (ikp1l1) [draw=red, fit= (np1ikp1l1) (d1ikp1l1) (1ikp1l1), inner sep=0.1cm, fill=red!20, fill opacity=0.2]{};
    \node[](d3ikp1l1)[below right=of d1ikp1l1]{$\dots$};

    \node[state](np1ikp1l2)[right=of 2nikl2]{$n+1$};
    \node[state](1ikp1l2)[above=of np1ikp1l2]{$1$};
    \node[](d1ikp1l2)[right=.2cm of 1ikp1l2]{$\dots$};
    \node[](d2ikp1l2)[right=.2cm of np1ikp1l2]{$\dots$};
    \node (ikp1l2) [draw=red, fit= (np1ikp1l2) (d1ikp1l2) (1ikp1l2), inner sep=0.1cm, fill=red!20, fill opacity=0.2]{};
    \node[](d3ikp1l2)[below right=of d1ikp1l2]{$\dots$};

    \node[](d1inl1)[above right=of d3ikp1l1]{$\dots$};
    \node[state](ninl1)[right=.2cm of d1inl1]{$n$};
    \node[state](2ninl1)[below=of ninl1]{$2n$};
    \node[](d2inl1)[left=.2cm of 2ninl1]{$\dots$};
    \node (inl1) [draw=red, fit= (d1inl1) (2ninl1) (ninl1), inner sep=0.1cm, fill=red!20, fill opacity=0.2]{};
    
    \node[](d1inl2)[above right=of d3ikp1l2]{$\dots$};
    \node[state](ninl2)[right=.2cm of d1inl2]{$n$};
    \node[state](2ninl2)[below=of ninl2]{$2n$};
    \node[](d2inl2)[left=.2cm of 2ninl2]{$\dots$};
    \node (inl2) [draw=red, fit= (d1inl2) (2ninl2) (ninl2), inner sep=0.1cm, fill=red!20, fill opacity=0.2]{};

    \node[state,accepting](fin)[right=.6cm of 2ninl1]{$1$};
    \node[state,accepting](fail)[right=.6cm of ninl2]{$0$};
    
    \path[->] 

    (start) edge [bend left=10,above left] node [align=center] {$q$} (1i1l1)
            edge [below left] node [align=center] {$1-q$} (1i1l2)
    
    (1ikl1) edge [above] node [align=center] {$p$} (2ikl1)
            edge [left] node [pos=.3] {$1-p$} (np2ikl1)
    (np1ikl1) edge [above] node [align=center] {$1-p$} (np2ikl1)
            edge [left] node [pos=.3] {$p$} (2ikl1)
    (nikl1) edge [left] node [align=center] {$1$} (2nikl1)
    (2nikl1) edge node {$q$} (1ikp1l1)
            edge node [pos=.7] {$1-q$} (1ikp1l2)

    (1ikl2) edge [above] node [align=center] {$p$} (2ikl2)
            edge [left] node [pos=.3] {$1-p$} (np2ikl2)
    (np1ikl2) edge [above] node [align=center] {$1-p$} (np2ikl2)
            edge [left] node [pos=.3] {$p$} (2ikl2)
    (2nikl2) edge [left] node [align=center] {$1$} (nikl2)

    (nikl2) edge node {$q$} (1ikp1l2)
            edge [bend left=3] node [pos=.4] {$1-q$} (1ikp1l1)

    (ninl1) edge [left] node [align=center] {$1$} (2ninl1)
    (2ninl1) edge node {$q$} (fin)
            edge node {$1-q$} (fail)

    (2ninl2) edge [left] node [align=center] {$1$} (ninl2)
    (ninl2) edge [] node [pos=.7] {$1-q$} (fin)
            edge [below] node {$q$} (fail)
    ;
    \end{tikzpicture}
}\qquad
\subfloat[Variant B. \label{fig:reach0_ptime_var_b}] {
    \begin{tikzpicture}[shorten >=1pt,auto,node distance=.9 cm, scale = 0.5, transform shape,initial text={}]
    
    \node[state, initial](start){$s$};

    \node[state](np1i1l1)[above right=of start]{$n+1$};
    \node[state](1i1l1)[above=of np1i1l1]{$1$};
    \node[](d1i1l1)[right=.2cm of 1i1l1]{$\dots$};
    \node[](d2i1l1)[right=.2cm of np1i1l1]{$\dots$};
    \node (i1l1) [draw=red, fit= (np1i1l1) (d1i1l1) (1i1l1), inner sep=0.1cm, fill=red!20, fill opacity=0.2]{};
    \node[](d3i1l1)[below right=of d1i1l1]{$\dots$};

    \node[state](1i1l2)[below right=of start]{$1$};
    \node[state](np1i1l2)[below=of 1i1l2]{$n+1$};
    \node[](d1i1l2)[right=.2cm of 1i1l2]{$\dots$};
    \node[](d2i1l2)[right=.2cm of np1i1l2]{$\dots$};
    \node (i1l2) [draw=red, fit= (np1i1l2) (d1i1l2) (1i1l2), inner sep=0.1cm, fill=red!20, fill opacity=0.2]{};
    \node[](d3i1l2)[below right=of d1i1l2]{$\dots$};

    \node[state](1ikl1)[right=2cm of d1i1l1]{$1$};
    \node[state](2ikl1)[right=of 1ikl1]{$2$};
    \node[](d1ikl1)[right=.2cm of 2ikl1]{$\dots$};
    \node[state](nikl1)[right=.3cm of d1ikl1]{$n$};
    \node[state](np1ikl1)[below=of 1ikl1]{$n+1$};
    \node[state](np2ikl1)[below=of 2ikl1]{$n+2$};
    \node[](d2ikl1)[right=.2cm of np2ikl1]{$\dots$};
    \node[state](2nikl1)[right=.2cm of d2ikl1]{$2n$};
    \node (ikl1) [draw=red, fit= (np1ikl1) (2nikl1) (nikl1), inner sep=0.1cm, fill=red!20, fill opacity=0.2]{};

    \node[state](1ikl2)[right=2cm of d1i1l2]{$1$};
    \node[state](2ikl2)[right=of 1ikl2]{$2$};
    \node[](d1ikl2)[right=.2cm of 2ikl2]{$\dots$};
    \node[state](nikl2)[right=.3cm of d1ikl2]{$n$};
    \node[state](np1ikl2)[below=of 1ikl2]{$n+1$};
    \node[state](np2ikl2)[below=of 2ikl2]{$n+2$};
    \node[](d2ikl2)[right=.2cm of np2ikl2]{$\dots$};
    \node[state](2nikl2)[right=.2cm of d2ikl2]{$2n$};
    \node (ikl1) [draw=red, fit= (np1ikl2) (2nikl2) (nikl2), inner sep=0.1cm, fill=red!20, fill opacity=0.2]{};

    \node[state](1ikp1l1)[right=of nikl1]{$1$};
    \node[state](np1ikp1l1)[below=of 1ikp1l1]{$n+1$};
    \node[](d1ikp1l1)[right=.2cm of 1ikp1l1]{$\dots$};
    \node[](d2ikp1l1)[right=.2cm of np1ikp1l1]{$\dots$};
    \node (ikp1l1) [draw=red, fit= (np1ikp1l1) (d1ikp1l1) (1ikp1l1), inner sep=0.1cm, fill=red!20, fill opacity=0.2]{};
    \node[](d3ikp1l1)[below right=of d1ikp1l1]{$\dots$};

    \node[state](np1ikp1l2)[right=of 2nikl2]{$n+1$};
    \node[state](1ikp1l2)[above=of np1ikp1l2]{$1$};
    \node[](d1ikp1l2)[right=.2cm of 1ikp1l2]{$\dots$};
    \node[](d2ikp1l2)[right=.2cm of np1ikp1l2]{$\dots$};
    \node (ikp1l2) [draw=red, fit= (np1ikp1l2) (d1ikp1l2) (1ikp1l2), inner sep=0.1cm, fill=red!20, fill opacity=0.2]{};
    \node[](d3ikp1l2)[below right=of d1ikp1l2]{$\dots$};

    \node[](d1inl1)[above right=of d3ikp1l1]{$\dots$};
    \node[state](ninl1)[right=.2cm of d1inl1]{$n$};
    \node[state](2ninl1)[below=of ninl1]{$2n$};
    \node[](d2inl1)[left=.2cm of 2ninl1]{$\dots$};
    \node (inl1) [draw=red, fit= (d1inl1) (2ninl1) (ninl1), inner sep=0.1cm, fill=red!20, fill opacity=0.2]{};
    
    \node[](d1inl2)[above right=of d3ikp1l2]{$\dots$};
    \node[state](ninl2)[right=.2cm of d1inl2]{$n$};
    \node[state](2ninl2)[below=of ninl2]{$2n$};
    \node[](d2inl2)[left=.2cm of 2ninl2]{$\dots$};
    \node (inl2) [draw=red, fit= (d1inl2) (2ninl2) (ninl2), inner sep=0.1cm, fill=red!20, fill opacity=0.2]{};

    \node[state,accepting](fin)[right=.6cm of 2ninl1]{$1$};
    \node[state,accepting](fail)[right=.6cm of ninl2]{$0$};
    
    \path[->] 

    (start) edge [bend left=10,above left] node [align=center] {$q$} (1i1l1)
            edge [below left] node [align=center] {$1-q$} (1i1l2)
    
    (1ikl1) edge [above] node [align=center] {$p$} (2ikl1)
            edge [left] node [pos=.3] {$1-p$} (np2ikl1)
    (np1ikl1) edge [above] node [align=center] {$1-p$} (np2ikl1)
            edge [left] node [pos=.3] {$p$} (2ikl1)
    (nikl1) edge [left] node [align=center] {$1$} (2nikl1)
    (2nikl1) edge node {$q$} (1ikp1l1)
            edge node [pos=.7] {$1-q$} (1ikp1l2)

    (1ikl2) edge [above] node [align=center] {$p$} (2ikl2)
            edge [left] node [pos=.3] {$1-p$} (np2ikl2)
    (np1ikl2) edge [above] node [align=center] {$1-p$} (np2ikl2)
            edge [left] node [pos=.3] {$p$} (2ikl2)
    (2nikl2) edge [left] node [align=center] {$1$} (nikl2)

    (nikl2) edge node {$1-q$} (1ikp1l2)
            edge [bend left=3] node [pos=.4] {$q$} (1ikp1l1)

    (ninl1) edge [left] node [align=center] {$1$} (2ninl1)
    (2ninl1) edge node {$q$} (fin)
            edge node {$1-q$} (fail)

    (2ninl2) edge [left] node [align=center] {$1$} (ninl2)
    (ninl2) edge [] node [pos=.7] {$q$} (fin)
            edge [below] node {$1-q$} (fail)
    ;
    \end{tikzpicture}
}\qquad
\subfloat[Variant C. \label{fig:reach0_ptime_var_c}] {
    \begin{tikzpicture}[shorten >=1pt,auto,node distance=.9 cm, scale = 0.5, transform shape,initial text={}]
    
    \node[state, initial](start){$s$};

    \node[state](np1i1l1)[above right=of start]{$n+1$};
    \node[state](1i1l1)[above=of np1i1l1]{$1$};
    \node[](d1i1l1)[right=.2cm of 1i1l1]{$\dots$};
    \node[](d2i1l1)[right=.2cm of np1i1l1]{$\dots$};
    \node (i1l1) [draw=red, fit= (np1i1l1) (d1i1l1) (1i1l1), inner sep=0.1cm, fill=red!20, fill opacity=0.2]{};
    \node[](d3i1l1)[below right=of d1i1l1]{$\dots$};

    \node[state](1i1l2)[below right=of start]{$1$};
    \node[state](np1i1l2)[below=of 1i1l2]{$n+1$};
    \node[](d1i1l2)[right=.2cm of 1i1l2]{$\dots$};
    \node[](d2i1l2)[right=.2cm of np1i1l2]{$\dots$};
    \node (i1l2) [draw=red, fit= (np1i1l2) (d1i1l2) (1i1l2), inner sep=0.1cm, fill=red!20, fill opacity=0.2]{};
    \node[](d3i1l2)[below right=of d1i1l2]{$\dots$};

    \node[state](1ikl1)[right=2cm of d1i1l1]{$1$};
    \node[state](2ikl1)[right=of 1ikl1]{$2$};
    \node[](d1ikl1)[right=.2cm of 2ikl1]{$\dots$};
    \node[state](nikl1)[right=.3cm of d1ikl1]{$n$};
    \node[state](np1ikl1)[below=of 1ikl1]{$n+1$};
    \node[state](np2ikl1)[below=of 2ikl1]{$n+2$};
    \node[](d2ikl1)[right=.2cm of np2ikl1]{$\dots$};
    \node[state](2nikl1)[right=.2cm of d2ikl1]{$2n$};
    \node (ikl1) [draw=red, fit= (np1ikl1) (2nikl1) (nikl1), inner sep=0.1cm, fill=red!20, fill opacity=0.2]{};

    \node[state](1ikl2)[right=2cm of d1i1l2]{$1$};
    \node[state](2ikl2)[right=of 1ikl2]{$2$};
    \node[](d1ikl2)[right=.2cm of 2ikl2]{$\dots$};
    \node[state](nikl2)[right=.3cm of d1ikl2]{$n$};
    \node[state](np1ikl2)[below=of 1ikl2]{$n+1$};
    \node[state](np2ikl2)[below=of 2ikl2]{$n+2$};
    \node[](d2ikl2)[right=.2cm of np2ikl2]{$\dots$};
    \node[state](2nikl2)[right=.2cm of d2ikl2]{$2n$};
    \node (ikl1) [draw=red, fit= (np1ikl2) (2nikl2) (nikl2), inner sep=0.1cm, fill=red!20, fill opacity=0.2]{};

    \node[state](1ikp1l1)[right=of nikl1]{$1$};
    \node[state](np1ikp1l1)[below=of 1ikp1l1]{$n+1$};
    \node[](d1ikp1l1)[right=.2cm of 1ikp1l1]{$\dots$};
    \node[](d2ikp1l1)[right=.2cm of np1ikp1l1]{$\dots$};
    \node (ikp1l1) [draw=red, fit= (np1ikp1l1) (d1ikp1l1) (1ikp1l1), inner sep=0.1cm, fill=red!20, fill opacity=0.2]{};
    \node[](d3ikp1l1)[below right=of d1ikp1l1]{$\dots$};

    \node[state](np1ikp1l2)[right=of 2nikl2]{$n+1$};
    \node[state](1ikp1l2)[above=of np1ikp1l2]{$1$};
    \node[](d1ikp1l2)[right=.2cm of 1ikp1l2]{$\dots$};
    \node[](d2ikp1l2)[right=.2cm of np1ikp1l2]{$\dots$};
    \node (ikp1l2) [draw=red, fit= (np1ikp1l2) (d1ikp1l2) (1ikp1l2), inner sep=0.1cm, fill=red!20, fill opacity=0.2]{};
    \node[](d3ikp1l2)[below right=of d1ikp1l2]{$\dots$};

    \node[](d1inl1)[above right=of d3ikp1l1]{$\dots$};
    \node[state](ninl1)[right=.2cm of d1inl1]{$n$};
    \node[state](2ninl1)[below=of ninl1]{$2n$};
    \node[](d2inl1)[left=.2cm of 2ninl1]{$\dots$};
    \node (inl1) [draw=red, fit= (d1inl1) (2ninl1) (ninl1), inner sep=0.1cm, fill=red!20, fill opacity=0.2]{};
    
    \node[](d1inl2)[above right=of d3ikp1l2]{$\dots$};
    \node[state](ninl2)[right=.2cm of d1inl2]{$n$};
    \node[state](2ninl2)[below=of ninl2]{$2n$};
    \node[](d2inl2)[left=.2cm of 2ninl2]{$\dots$};
    \node (inl2) [draw=red, fit= (d1inl2) (2ninl2) (ninl2), inner sep=0.1cm, fill=red!20, fill opacity=0.2]{};

    \node[state,accepting](fin)[right=.6cm of 2ninl1]{$1$};
    \node[state,accepting](fail)[right=.6cm of ninl2]{$0$};
    
    \path[->] 

    (start) edge [bend left=10,above left] node [align=center] {$q$} (1i1l1)
            edge [below left] node [align=center] {$1-q$} (1i1l2)
    
    (1ikl1) edge [above] node [align=center] {$p$} (2ikl1)
            edge [left] node [pos=.3] {$1-p$} (np2ikl1)
    (np1ikl1) edge [above] node [align=center] {$p$} (np2ikl1)
            edge [left] node [pos=.3] {$1-p$} (2ikl1)
    (nikl1) edge [left] node [align=center] {$1$} (2nikl1)
    (2nikl1) edge node {$q$} (1ikp1l1)
            edge node [pos=.7] {$1-q$} (1ikp1l2)

    (1ikl2) edge [above] node [align=center] {$p$} (2ikl2)
            edge [left] node [pos=.3] {$1-p$} (np2ikl2)
    (np1ikl2) edge [above] node [align=center] {$p$} (np2ikl2)
            edge [left] node [pos=.3] {$1-p$} (2ikl2)
    (nikl2) edge [left] node [align=center] {$1$} (2nikl2)
    (2nikl2) edge node {$q$} (1ikp1l2)
            edge [bend left=3] node [pos=.4] {$1-q$} (1ikp1l1)

    (ninl1) edge [left] node [align=center] {$1$} (2ninl1)
    (2ninl1) edge node {$p$} (fin)
            edge node {$1-p$} (fail)

    (ninl2) edge [left] node [align=center] {$1$} (2ninl2)
    (2ninl2) edge [bend right=22] node [pos=.7] {$p$} (2ninl1)
            edge [right] node {$1-p$} (fail)
    ;
    \end{tikzpicture}
}\qquad
\subfloat[Variant D. \label{fig:reach0_ptime_var_d}] {
    \begin{tikzpicture}[shorten >=1pt,auto,node distance=.9 cm, scale = 0.5, transform shape,initial text={}]
    
    \node[state, initial](start){$s$};

    \node[state](np1i1l1)[above right=of start]{$n+1$};
    \node[state](1i1l1)[above=of np1i1l1]{$1$};
    \node[](d1i1l1)[right=.2cm of 1i1l1]{$\dots$};
    \node[](d2i1l1)[right=.2cm of np1i1l1]{$\dots$};
    \node (i1l1) [draw=red, fit= (np1i1l1) (d1i1l1) (1i1l1), inner sep=0.1cm, fill=red!20, fill opacity=0.2]{};
    \node[](d3i1l1)[below right=of d1i1l1]{$\dots$};

    \node[state](1i1l2)[below right=of start]{$1$};
    \node[state](np1i1l2)[below=of 1i1l2]{$n+1$};
    \node[](d1i1l2)[right=.2cm of 1i1l2]{$\dots$};
    \node[](d2i1l2)[right=.2cm of np1i1l2]{$\dots$};
    \node (i1l2) [draw=red, fit= (np1i1l2) (d1i1l2) (1i1l2), inner sep=0.1cm, fill=red!20, fill opacity=0.2]{};
    \node[](d3i1l2)[below right=of d1i1l2]{$\dots$};

    \node[state](1ikl1)[right=2cm of d1i1l1]{$1$};
    \node[state](2ikl1)[right=of 1ikl1]{$2$};
    \node[](d1ikl1)[right=.2cm of 2ikl1]{$\dots$};
    \node[state](nikl1)[right=.3cm of d1ikl1]{$n$};
    \node[state](np1ikl1)[below=of 1ikl1]{$n+1$};
    \node[state](np2ikl1)[below=of 2ikl1]{$n+2$};
    \node[](d2ikl1)[right=.2cm of np2ikl1]{$\dots$};
    \node[state](2nikl1)[right=.2cm of d2ikl1]{$2n$};
    \node (ikl1) [draw=red, fit= (np1ikl1) (2nikl1) (nikl1), inner sep=0.1cm, fill=red!20, fill opacity=0.2]{};

    \node[state](1ikl2)[right=2cm of d1i1l2]{$1$};
    \node[state](2ikl2)[right=of 1ikl2]{$2$};
    \node[](d1ikl2)[right=.2cm of 2ikl2]{$\dots$};
    \node[state](nikl2)[right=.3cm of d1ikl2]{$n$};
    \node[state](np1ikl2)[below=of 1ikl2]{$n+1$};
    \node[state](np2ikl2)[below=of 2ikl2]{$n+2$};
    \node[](d2ikl2)[right=.2cm of np2ikl2]{$\dots$};
    \node[state](2nikl2)[right=.2cm of d2ikl2]{$2n$};
    \node (ikl1) [draw=red, fit= (np1ikl2) (2nikl2) (nikl2), inner sep=0.1cm, fill=red!20, fill opacity=0.2]{};

    \node[state](1ikp1l1)[right=of nikl1]{$1$};
    \node[state](np1ikp1l1)[below=of 1ikp1l1]{$n+1$};
    \node[](d1ikp1l1)[right=.2cm of 1ikp1l1]{$\dots$};
    \node[](d2ikp1l1)[right=.2cm of np1ikp1l1]{$\dots$};
    \node (ikp1l1) [draw=red, fit= (np1ikp1l1) (d1ikp1l1) (1ikp1l1), inner sep=0.1cm, fill=red!20, fill opacity=0.2]{};
    \node[](d3ikp1l1)[below right=of d1ikp1l1]{$\dots$};

    \node[state](np1ikp1l2)[right=of 2nikl2]{$n+1$};
    \node[state](1ikp1l2)[above=of np1ikp1l2]{$1$};
    \node[](d1ikp1l2)[right=.2cm of 1ikp1l2]{$\dots$};
    \node[](d2ikp1l2)[right=.2cm of np1ikp1l2]{$\dots$};
    \node (ikp1l2) [draw=red, fit= (np1ikp1l2) (d1ikp1l2) (1ikp1l2), inner sep=0.1cm, fill=red!20, fill opacity=0.2]{};
    \node[](d3ikp1l2)[below right=of d1ikp1l2]{$\dots$};

    \node[](d1inl1)[above right=of d3ikp1l1]{$\dots$};
    \node[state](ninl1)[right=.2cm of d1inl1]{$n$};
    \node[state](2ninl1)[below=of ninl1]{$2n$};
    \node[](d2inl1)[left=.2cm of 2ninl1]{$\dots$};
    \node (inl1) [draw=red, fit= (d1inl1) (2ninl1) (ninl1), inner sep=0.1cm, fill=red!20, fill opacity=0.2]{};
    
    \node[](d1inl2)[above right=of d3ikp1l2]{$\dots$};
    \node[state](ninl2)[right=.2cm of d1inl2]{$n$};
    \node[state](2ninl2)[below=of ninl2]{$2n$};
    \node[](d2inl2)[left=.2cm of 2ninl2]{$\dots$};
    \node (inl2) [draw=red, fit= (d1inl2) (2ninl2) (ninl2), inner sep=0.1cm, fill=red!20, fill opacity=0.2]{};

    \node[state,accepting](fin)[right=.6cm of 2ninl1]{$1$};
    \node[state,accepting](fail)[right=.6cm of ninl2]{$0$};
    
    \path[->] 

    (start) edge [bend left=10,above left] node [align=center] {$q$} (1i1l1)
            edge [below left] node [align=center] {$1-q$} (1i1l2)
    
    (1ikl1) edge [above] node [align=center] {$p$} (2ikl1)
            edge [left] node [pos=.3] {$1-p$} (np2ikl1)
    (np1ikl1) edge [above] node [align=center] {$p$} (np2ikl1)
            edge [left] node [pos=.3] {$1-p$} (2ikl1)
    (nikl1) edge [left] node [align=center] {$1$} (2nikl1)
    (2nikl1) edge node {$1-q$} (1ikp1l1)
            edge node [pos=.7] {$q$} (1ikp1l2)

    (1ikl2) edge [above] node [align=center] {$p$} (2ikl2)
            edge [left] node [pos=.3] {$1-p$} (np2ikl2)
    (np1ikl2) edge [above] node [align=center] {$p$} (np2ikl2)
            edge [left] node [pos=.3] {$1-p$} (2ikl2)
    (nikl2) edge [left] node [align=center] {$1$} (2nikl2)
    (2nikl2) edge node {$q$} (1ikp1l2)
            edge [bend left=3] node [pos=.4] {$1-q$} (1ikp1l1)

    (ninl1) edge [left] node [align=center] {$1$} (2ninl1)
    (2ninl1) edge node {$p$} (fin)
            edge node {$1-p$} (fail)

    (ninl2) edge [left] node [align=center] {$1$} (2ninl2)
    (2ninl2) edge [bend right=22] node [pos=.7] {$p$} (2ninl1)
            edge [right] node {$1-p$} (fail)
    ;
    \end{tikzpicture}
}
    \caption{The 4 variants of our benchmark for a given parameter $n$ which have $4n^2+3$ states and $2n+3$ equivalence classes each. Every red box represents one equivalence class with $2n$ states.}
    \label{fig:custom}
\end{figure}
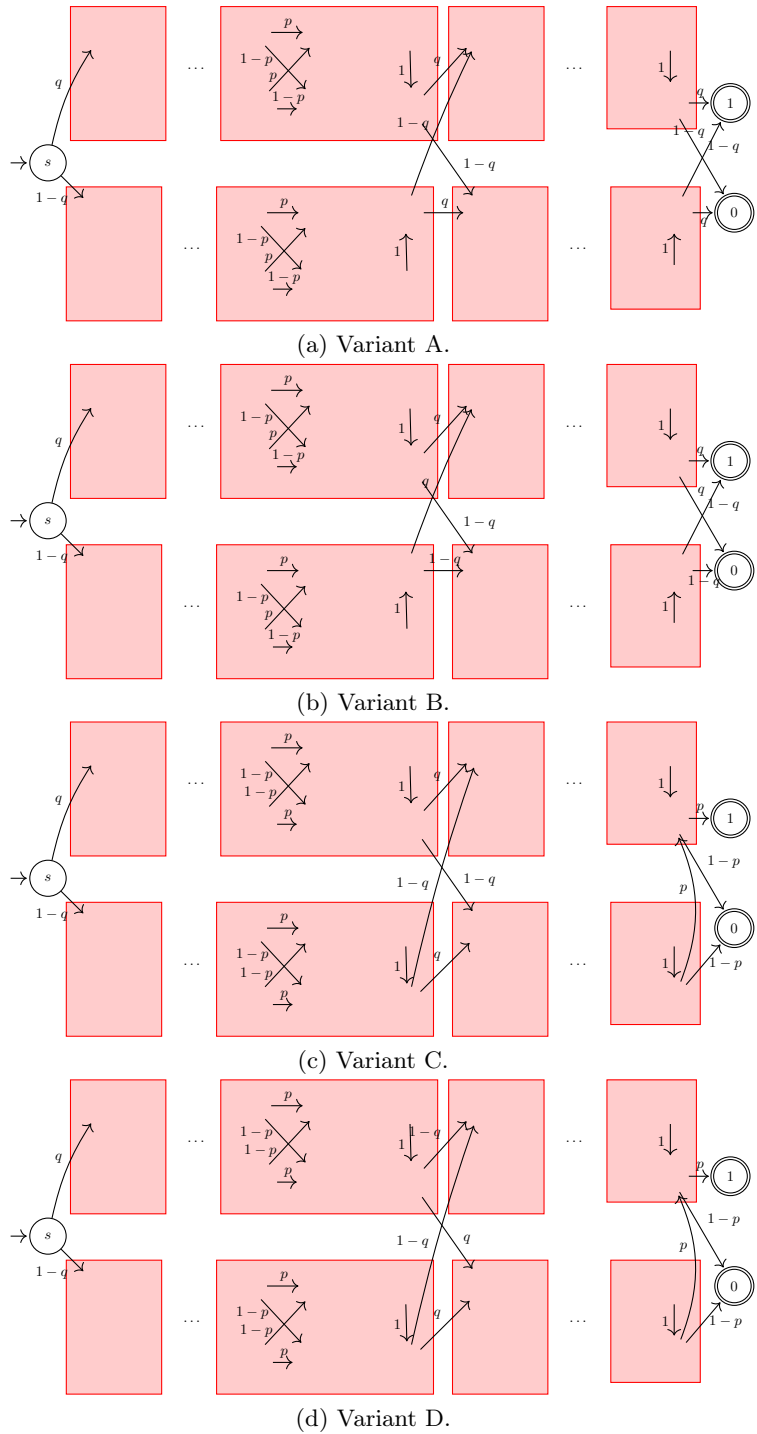

We implemented \Cref{algo:eq} in Python and set up an experimental pipeline as depicted in \Cref{fig:experimental_flow_chart}. First each model, specified in PRISM format, is read by Storm-pars (v1.8.1) and converted to the JSON explicit format (A). This explicit model is then reduced in size using our algorithm (B). Then, the explicit representations of both the original and collapsed models are converted back into PRISM (C). The latter is done by explicitly hard-coding every transition. Finally, both the original and collapsed models are used as input to either the monotonicity or parameter lifting module of Storm-pars (D).

All experiments were run on a 2021 MacBook M1 Pro. For Storm, this we used the docker image \texttt{movesrwth/storm:1.8.1} with 8GB allocated memory. When reproducing the monotonicity  experiments from \cite{monotonicity}, a time-out of 4h was used; the parameter-lifting experiments from \cite{DBLP:conf/tacas/SpelJK21} were given a 30min timeout. The bounds were chosen to match the experimental setup of the cited papers.

We should note that we only ran the monotonicy analysis and parameter-lifting modules of Storm on the benchmarks that actually reported a size reduction. Those benchmarks for which there was no size reduction were not considered further. The excluded benchmarks include the Gambler, Zeroconf, and Message authentication benchmarks \cite{DBLP:conf/tacas/SpelJK21} as well as the Herman benchmark \cite{DBLP:conf/vmcai/HeckSJMK22}.


\subsection{Custom benchmarks}

In order to more accurately evaluate the performance of our algorithm we have also constructed our own benchmark with 4 variants. These are depicted in \Cref{fig:reach0_ptime_var_a,fig:reach0_ptime_var_b,fig:reach0_ptime_var_c,fig:reach0_ptime_var_d}. All of these variants consist of a start state $s$, the final state $1$, the fail state $0$ and $2n$ red ``blocks'' with each block depicting an equivalence class containing $2n$ states. Hence, the total number of states in each block is $4n^2+3$, although in our implementation the states in the beginning of each block with no incoming edges (for instance, the state $n+1$ in the blocks in \Cref{fig:reach0_ptime_var_a}) are automatically deleted. This leads to the actual number of states in the model being $4n^2-2n+1$. These variants all share the same structure of the states, but differ slightly in the transitions between them. It is shown in \Cref{tab:custom_results} that these different benchmarks induce different behavior in the monotonicity analysis.

All variants have an equal number of states, and also have two parameters $p$ and $q$. The size of the model can be adjusted using the variable $n$. In variants A and B the parameter $p$ is no longer present after applying our algorithm, leaving only $q$ since the only transitions with the parameter $p$ are present within the blocks, and these are collapsed into a single state. The variants C and D retain both parameters $p$ and $q$ after applying our algorithm since $p$ also occurs in the transitions from the last two blocks to the final and fail states.

The monotonicity properties of the four variants also differ. Variants A and C are not monotonic in $q$ while variants B and D are monotonically increasing in $q$. Moreover, C and D are both monotonically increasing in $p$.

\begin{table}[]
    \centering
    \begin{tabular}{|c|c|c|c|c|c|c|c|c|}
    \hline
        Benchmark & \multicolumn{2}{|c|}{Size (states)} & Collapse time (s) & \multicolumn{2}{|c|}{Mono. time (s)} & \multicolumn{2}{|c|}{Lifting time (s)} \\
    \hline
         & Before & after &  & before & after & before & after\\
    \hline
    \hline
      Custom (A) & & & & & & & \\
    \hline
      n=2 & 15 & 7 & <1 & <1 & <1 & <1 & <1 \\
    \hline
      n=3 & 33 & 9 & <1 & <1 & <1 & <1 & <1  \\
    \hline
      n=8 & 243 & 19 & <1 & 4.38 & <1 & <1 & <1  \\
    \hline
      n=10 & 383 & 23 & <1 & TO & <1 & <1 & <1  \\
    \hline
      n=15 & 873 & 33 & <1 & TO & <1 & <1 & <1  \\
    \hline
      n=50 & 9903 & 103 & 3.52 & TO & <1 & 5.781 & <1 \\
    \hline
      n=100 & 39803 & 203 & 28.26 & N/A & N/A & 87.535 & <1 \\
    \hline
      n=150 & 89703 & 303 & 88.78 & N/A & N/A & TO (file) & <1\\
    \hline
    \hline
      Custom (B) & & & & & & &  \\
    \hline
      n=2 & 15 & 7 & <1 & <1 & <1 & <1 & <1  \\
    \hline
      n=25 & 2453 & 53 & <1 & <1 & <1 & <1 & <1  \\
    \hline
      n=50 & 9903 & 103 & 3.29 & <1 & <1 & <1 & <1 \\
    \hline
      n=100 & 39803 & 203 & 25.71 & N/A & N/A & <1 & <1 \\
    \hline
      n=150 & 89703 & 303 & 87.76 & N/A & N/A & TO (file) & <1 \\
    \hline
    \hline
      Custom (C) & & & & & & &  \\
    \hline
      n=2 & 15 & 7 & <1 & <1 & <1 & <1 & <1  \\
    \hline
      n=3 & 33 & 9 & <1 & <1 & <1 & <1 & <1  \\
    \hline
      n=8 & 243 & 19 & <1 & 2.24 & <1 & <1 & <1  \\
    \hline
      n=10 & 383 & 23 & <1 & 30.28 & <1 & <1 & <1  \\
    \hline
      n=15 & 873 & 33 & <1 & 10381.17 & <1 & <1 & <1  \\
    \hline
      n=20 & 1563 & 43 & <1 & TO & <1 & <1 & <1  \\
    \hline
      n=50 & 9903 & 103 & 3.29 & TO & <1 & <1 & <1 \\
    \hline
      n=100 & 39803 & 203 & 25.60 & N/A & N/A & <1 & <1 \\
    \hline
      n=150 & 89703 & 303 & 86.36 & N/A & N/A & TO (file) & <1 \\
    \hline
    \hline
      Custom (D) & & & & & & &  \\
    \hline
      n=2 & 15 & 7 & <1 & <1 & <1 & <1 & <1  \\
    \hline
      n=25 & 2453 & 53 & <1 & <1 & <1 & <1 & <1  \\
    \hline
      n=50 & 9903 & 103 & 3.35 & <1 & <1 &  <1 & <1\\
    \hline
      n=100 & 39803 & 203 & 25.73 & N/A & N/A & 1.330 & <1\\
    \hline
      n=150 & 89703 & 303 & 89.31 & N/A & N/A & TO (file) & <1 \\
    \hline
    \end{tabular}
    \caption{Comparison of sizes and times to check monotonicity and realize parameter lifting with and without our reduction techniques. We write TO to denote time-outs, N/A when the experiment was not executed,  and TO (file) when the module ran out of time while parsing the file. (Recall we are using explicitly encoded transitions.)}
    \label{tab:custom_results}
\end{table}

\subsection{Results and tables}

A quantitative summary of the results can be found in \Cref{tab:mono_benchmarks,tab:lifting_benchmarks,tab:custom_results,tab:custom_big}. Experiments that were not performed are denoted as N/A. The benchmarks are listed together with the constant values used when instantiating the models. 

\Cref{tab:custom_results} contains the results of applying our algorithm, as well as the monotonicity analysis and parameter lifting modules of Storm. The first column contains the details of the benchmarks. Next is the size of the model before and after applying our algorithm. This is followed by the time our algorithm took to identify and collapse the equivalence classes. The last four columns contain the running times of the Storm-pars modules for monotonicity analysis, and parameter lifting. Times are listed with and without applying our algorithm before the model is passed to Storm.
The reported time of the collapse algorithm only includes the time it took to identify and collapse the equivalence classes, and does not include parsing or collapsing of the extremal states; the values for monotonicity-analysis and parameter-lifting times are those reported by Storm as ``monotonicity analysis'' and ``model checking'', respectively.

In \Cref{tab:custom_results}, we observe that there is a significant difference in the running time of the monotonicity analysis for those benchmarks that have a non-monotonic parameter (i.e., variants A and C). For parameter lifting, there are differences in running times for bigger models ($n \geq 50$) coming from benchmarks that contain a non-monotonic parameter.

For some values of benchmark constants, we did not run the monotonocity analysis or parameter lifting, since we expected time-outs or memory-outs. We did, however, run our reduction algorithm on those models, with the results visible in \Cref{tab:custom_big}. We observe it is possible to run our algorithm on models of up to a million states in a reasonable amount of time.

\begin{table}[ht]
    \centering
    \begin{tabular}{|c|c|c|c|c|}
    \hline
        Benchmark & \multicolumn{2}{|c|}{Size (states)} & Collapse time (s) \\
    \hline
         & before & after &  \\
    \hline
    \hline
      Custom (A) & & & \\
    \hline
      n=300 & 359403 & 603 & 721.19 \\
    \hline
      n=400 & 639203 & 803 & 1647.76 \\
    \hline
      n=500 & 999003 & 1003 & 3162.41\\
    \hline
    \hline
      Custom (B) & & & \\
    \hline
      n=300 & 359403 & 603 & 734.91 \\
    \hline
      n=400 & 639203 & 803 & 1634.48 \\
    \hline
      n=500 & 999003 & 1003 & 3106.48 \\
    \hline
    \hline
      Custom (C) & & & \\
    \hline
    n=300 & 359403 & 603 & 709.35 \\
    \hline
      n=400 & 639203 & 803 & 1630.56  \\
    \hline
      n=500 & 999003 & 1003 & 3095.20  \\
    \hline
    \hline
      Custom (D) & & & \\
    \hline
      n=300 & 359403 & 603 & 699.97 \\
    \hline
      n=400 & 639203 & 803 & 1624.82  \\
    \hline
      n=500 & 999003 & 1003 & 3050.54 \\
    \hline
    \end{tabular}
    \caption{Comparison of size and reduction times only, for larger instances of our custom benchmarks.}
    \label{tab:custom_big}
\end{table}

Aside from running our algorithm on our custom benchmarks, we also performed experiments on the benchmarks from Spel et al. \cite{DBLP:conf/tacas/SpelJK21} as well as Heck et al. \cite{DBLP:conf/vmcai/HeckSJMK22}. The results can be seen in \Cref{tab:mono_benchmarks,tab:lifting_benchmarks}, respectively. The meaning of the columns is the same as in \Cref{tab:custom_results}. For these benchmarks, our algorithm causes no noticeable differences in the running times of the Storm-pars modules. The model size is significantly reduced, however.

\begin{table}[h]
    \centering
    \begin{tabular}{|c|c|c|c|c|c|c|}
    \hline
        Benchmark & \multicolumn{2}{|c|}{Size (states)} & Collapse time (s) & \multicolumn{2}{|c|}{Mono. time (s)} \\
    \hline
         & before & after &  & before & after\\
    \hline
    \hline
      brp & & & & &  \\
    \hline
      MAX=2, N=16 & 494 & 192 & <1 & <1 & <1   \\
    \hline
      MAX=2, N=32 & 990 & 384 & 1.02 & <1 & <1 \\
    \hline
      MAX=4, N=16 & 906 & 350 & <1 & <1 & <1   \\
    \hline
      MAX=4, N=32 & 1818 & 702 & 3.47 & <1 & <1 \\
    \hline
    \hline
      Crowds & & & & & \\
    \hline
      Size=5, Runs=3 & 268 & 182 & <1 & 82.19 & 71.66  \\
    \hline
      Size=5, Runs=6 & 6905 & 3782 & 79.10 & TO & TO  \\
    \hline 
    \end{tabular}
    \caption{Comparison of sizes and times to check monotonicity with and without our reduction technique for some benchmarks from the original monotonicity study.}
    \label{tab:mono_benchmarks}
\end{table}

\begin{table}[h]
    \centering
    \begin{tabular}{|c|c|c|c|c|c|c|}
    \hline
        Benchmark & \multicolumn{2}{|c|}{Size (states)} & Collapse time (s) & \multicolumn{2}{|c|}{Lifting time (s)} \\
    \hline
         & Before & after &  & before & after\\
    \hline
    \hline
      NRP & & & & &  \\
    \hline
      5,1 & 33 & 18 & <1 & 33.26 & 33.58  \\
    \hline
      6,1 & 45 & 24 & <1 & 1028.81 & 1029.974 \\
    \hline
      7,1 & 59 & 31 & <1 &  TO & MO \\
    \hline
      8,1 & 75 & 39 & <1 & TO & TO \\
    \hline
      9,1 & 93 & 48 & <1 &  TO & MO \\
    \hline
    \end{tabular}
    \caption{Comparison of sizes and times to realize parameter lifting with and without our reduction technique for the NRP benchmark from the original parameter lifting. We write MO  when the memory ran out.}
    \label{tab:lifting_benchmarks}
\end{table}

\section{Conclusions}
We have established a reduction from the NWR problem to the monotonicity problem. On the way, we used arithmetic circuits to restate and reformulate some known results about the computation of value functions. This new approach allowed us to obtain a pMC representation of the partial derivatives of the value functions of a given pMC. This representation turned out to be quasi-polynomial. The question of whether such a partial-derivative pMC of polynomial size can always be constructed is left unanswered.

In a more practical direction, we took the equivalence-class detection algorithm from \cite{DBLP:conf/atva/EngelenPR23} and presented it as a pre-processing step for monotonicity analysis and parameter lifting. To evaluate the idea, we implemented the algorithm and realized some experiments on old and new benchmarks. Our results point to the algorithm being useful in reducing the size of (our custom) benchmarks. Unsurprisingly, when our algorithm succeeds in drastically reducing the size of a (custom) benchmark, monotonicity analysis and parameter lifting do benefit from a performance boost by working on a smaller model. Unfortunately, for benchmarks introduced in previous works, while we do see some reduction in size by using our algorithm, neither monotonicity analysis nor parameter lifting seem to run faster in the resulting model. Based on these, we believe it may be a good idea to implement the pre-processing step as an option within the Storm model checker.

\subsection*{Acknowledgements}
We thank Micha\"el Cadilhac and Nikhil Balaji for pointers to the
arithmetic-circuit literature. We are also grateful to Linus Heck for helping us with Storm.

\bibliographystyle{splncs04}
\bibliography{refs}

\end{document}